\documentclass[a4paper,british,cleveref, autoref, thm-restate]{article}
\usepackage{amssymb,amsmath,multicol}
\usepackage{graphicx,url}
\usepackage{color,setspace,enumitem}

\usepackage[ruled]{algorithm}
\usepackage[noend]{algpseudocode}
\usepackage[normalem]{ulem}

\usepackage{xspace}

\newtheorem{theorem}{Theorem}
\newtheorem{lemma}{Lemma}
\newtheorem{definition}{Definition}
\newenvironment{proof}{\noindent{\bf Proof.}}{\hfill$\Box$}

\usepackage{etoolbox}\AtBeginEnvironment{algorithmic}{\small}
\floatname{algorithm}{{\small Code}}

\algblockdefx[Receive]{Receive}{EndReceive}%
[1]{{\bf receive}   (#1) from process $\pr$}%
{{\bf end receive}}

\algblockdefx[Receivex]{Receivex}{EndReceivex}%
[2]{{\bf receive}   (#1) from process #2}%
{{\bf end receive}}

\algblockdefx[Upon]{Upon}{EndUpon}%
[1]{{\bf upon} (#1) do}%
{{\bf end upon}}

\makeatletter
\ifthenelse{\equal{\ALG@noend}{t}}%
{\algtext*{EndReceive}}
{}%

\makeatletter
\ifthenelse{\equal{\ALG@noend}{t}}%
{\algtext*{EndUpon}}
{}%

\algrenewcommand{\ALG@beginalgorithmic}{\small}
\algrenewcommand\alglinenumber[1]{\small #1:}

\makeatother

\usepackage{wrapfig}

\newcommand{\floor}[1]{\left\lfloor #1 \right\rfloor}

\providecommand{\tup}[1]{%
    \relax\ifmmode
      \langle #1 \rangle%
    \else
        $\langle$#1$\rangle$%
    \fi
}

\newcommand{\act}[1]{%
    \relax\ifmmode
        \mathord{\mathcode`\-="702D\sf #1\mathcode`\-="2200}%
    \else
        $\mathord{\mathcode`\-="702D\sf #1\mathcode`\-="2200}$%
    \fi
}

\newcommand{\remove}[1]{}

\makeatletter
\def\mainlistofsymbols{
  \normalsize
  \vspace*{1.5 em}
  \@starttoc{los}
}

\def\partonelistofsymbols{
  \normalsize
  \vspace*{1.5 em}
  \@starttoc{p1los}
}

\def\parttwolistofsymbols{
  \normalsize
  \vspace*{1.5 em}
  \@starttoc{p2los}
}

\def\l@symbol#1#2{\addpenalty{-\@highpenalty} \vskip 4pt plus 2pt
{\@dottedtocline{0}{0em}{8em}{#1}{#2}}}
\makeatother

\newcommand{\newhiddensym}[2]{%
}

\newcommand{\algIOA}[2]{\ifmmode{\text{#1}_{#2}}\else{$\text{#1}_{#2}$}\fi}

\newcommand{\EX}{\ifmmode{\xi}\else{$\xi$}\fi}
\newcommand{\EXF}{\ifmmode{\phi}\else{$\phi$}\fi}

\newcommand{\hist}[1]{H_{#1}}

\newcommand{\obj}[1]{O_{#1}}

\newcommand{\inter}[1]{
	\ifmmode{\left(\bigcap_{\mathcal{Q}\in#1}\mathcal{Q}\right)}
	\else{$\left(\bigcap_{\mathcal{Q}\in#1}\mathcal{Q}\right)$}
	\fi
}

\newcommand{\ledger}{\mathcal{GS}}

\newcommand{\op}{\pi}

\mathchardef\mhyphen="2D

\newcommand{\pr}{p}
\newcommand{\rdr}{r}

\newcommand{\bef}{\rightarrow}

\newcommand{\vid}[1]{\ifmmode{\nu_{#1}}\else{$\nu_{#1}$}\fi}

\newcommand{\seen}{\ifmmode{seen}\else{$seen$}\fi}

\newcommand{\maxts}[1]{\ifmmode{maxTS_{#1}}\else{$maxTS_{#1}$}\fi}
\newcommand{\maxtag}[1]{\ifmmode{maxTag_{#1}}\else{$maxTag_{#1}$}\fi}
\newcommand{\maxpair}[1]{\ifmmode{maxMPair_{#1}}\else{$maxMPair_{#1}$}\fi}
\newcommand{\mintag}[1]{\ifmmode{minTag_{#1}}\else{$minTag_{#1}$}\fi}
\newcommand{\maxps}{\ifmmode{maxPS}\else{$maxPS$}\fi}
\newcommand{\conftg}[1]{\ifmmode{confirmed_{#1}}\else{$confirmed_{#1}$}\fi}
\newcommand{\maxconftag}{\ifmmode{\ms{maxCT}}\else{$maxCT$}\fi}
\newcommand{\append}{{\textsc{append}\xspace}}

\newcommand{\atomicappends}{Atomic~Appends\xspace}

\newcommand{\gset}{G-Set\xspace}
\newcommand{\DGSO}{DSO\xspace}
\newcommand{\BDGSO}{BDSO\xspace}

\newcommand{\extends}{\Vert}

\newcommand{\ar}[1]{\textcolor{cyan}{#1}}
\newcommand{\red}[1]{{\textcolor{red}{#1}}}
\newcommand{\blue}[1]{{\textcolor{blue}{#1}}}
\newcommand{\af}[1]{{\textcolor{magenta}{#1}}}

\title{
Byzantine-tolerant Distributed Grow-only Sets: Specification and Applications\thanks{Partially supported by Regional Government of Madrid (CM) grant EdgeData-CM (P2018/TCS4499, cofunded by FSE \& FEDER) and Spanish Ministry of Science and Innovation grant ECID (PID2019-109805RB-I00, cofunded by FEDER).}
}

\author{Vicent Cholvi\\
Universitat Jaume I, Spain
\and
Antonio Fern\'andez Anta\\
IMDEA Networks Institute, Spain
\and
Chryssis Georgiou\\
Dept. of Computer Science, University of Cyprus, Cyprus
\and
Nicolas Nicolaou\\
Algolysis Ltd., Cyprus
\and
Michel Raynal\thanks{Supported by French ANR project ByBLoS (ANR-20-CE25-0002-01).}\\
IRISA, France \& PolyU, Hong Kong
\and
Antonio Russo\\
IMDEA Networks Institute, Spain}

\begin{document}

\maketitle 

\setcounter{footnote}{0}

\begin{abstract}
In order to formalize Distributed Ledger Technologies and their interconnections, a recent line of research work has formulated the notion of Distributed Ledger Object (DLO), which is a concurrent object that maintains a totally ordered sequence of records, abstracting blockchains and distributed ledgers. Through DLO, the Atomic Appends problem, intended as the need of a primitive able to append multiple records to distinct ledgers in an atomic way, is studied as a basic interconnection problem among ledgers.

In this work, we propose the {\em Distributed Grow-only Set object} (\DGSO), which instead of maintaining a sequence of records, as in a DLO, maintains a set of records in an immutable way: only Add and Get operations are provided. This object is inspired by the Grow-only Set (G-Set) data type which is part of the Conflict-free Replicated Data Types. We formally specify the object and we provide a consensus-free Byzantine-tolerant implementation that guarantees eventual consistency. We then use our Byzantine-tolerant \DGSO (\BDGSO) implementation to provide consensus-free algorithmic solutions to the Atomic Appends and {Atomic Adds} (the analogous problem of atomic appends applied on G-Sets) problems, as well as to construct consensus-free Single-Writer BDLOs.  
We believe that the \BDGSO\ has applications beyond the above-mentioned problems.
\end{abstract}

\section{Introduction}
\label{sec:Intro}
Blockchains (as termed by Nakamoto in \cite{N08bitcoin}) or Distributed Ledger Technologies (DLTs) (as used in \cite{DLO_SIGACT18} and \cite{Raynal18}) became one of the most trendy data structures 
following the introduction of crypto-currencies \cite{N08bitcoin}
and their recent application in finance and token-economy. Despite their early wide adoption, little was known initially about the 
fundamental construction and semantic properties of DLTs. 
A number of research groups 
attempted to provide rigorous definitions to characterise the fundamental properties
of DTLs as those used in Bitcoin and beyond \cite{AnceaumePLPP19, DLO_SIGACT18, GKL15}. Among those, Fern\'andez Anta et al. \cite{DLO_SIGACT18},
was the first to identify and provide a formal definition of a reliable concurrent object, termed \emph{Distributed Ledger Object} (DLO), which conveys the essential 
building block for many DLTs. In particular, a DLO maintains a sequence of records,
and supports two basic operations: \act{append} and \act{get}. The \act{append} operation 
is used to add a new record at the end of the sequence, while the \act{get} operation
returns the whole sequence. Implementations of DLOs under client and server crashes were proposed in~\cite{DLO_SIGACT18}, and under Byzantine failures in~\cite{BAA2020}. 

The introduction to many different DLT systems have led multiple studies \cite{BAA2020, AA2019, DBLP:conf/podc/Herlihy18, Koens2019} to investigate 
the possibility of DLT interoperability, i.e., the ability for an action to be applied
over a set of DLTs, rather than in a single DTL at a time. 
Using the DLO formalism, 
\cite{AA2019} introduced the {\em \atomicappends{} problem}, 
in which several clients have a ``composite''
record (a set of semantically-linked ``basic'' records) to append. Each basic record has to be appended to a different DLO, and it must
be guaranteed that either all basic records are appended to their DLOs
or none of them is appended. 

Consider, for example, two clients $A$
and $B$, where $A$ buys a car from $B$. Record $r_A$ includes the
transfer of the car's digital deed from $B$ to $A$, and $r_B$ includes
the transfer from $A$ to $B$ of the agreed amount in some digital
currency. DLO$_A$ is a ledger maintaining digital deeds and DLO$_B$
maintains transactions in some pre-agreed digital currency.  So, while
the two records are mutually dependent, they concern different DLOs.
Hence, the \atomicappends{} problem requires that {\em either} record
$r_A$ is appended in DLO$_A$ and record $r_B$ is appended in DLO$_B$,
{\em or} no record is appended in the corresponding {DLOs}. 

In the work presented in~\cite{AA2019}, 
the authors assumed that clients may fail by crashing and showed that for some cases the
{existence of an intermediary is necessary.
They materialized such an intermediary by implementing a specialized
DLT, termed \emph{Smart} DLO (SDLO).  Using the SDLO, the authors
solved the \atomicappends{} problem
in a client competitive asynchronous  environment, in which any number
of clients, and up to $f$ servers implementing the DLOs, may crash.}
%
A subsequent work solved the problem assuming Byzantine failures~\cite{BAA2020}, 
by introducing the notion of \emph{Byzantine Distributed Ledger Objects} (BDLO). 
Solutions for implementing BDLOs were presented, with each solution
relying on an underlying Byzantine Total-order Broadcast Service (BToB) \cite{coelho2018byzantine,CRISTIAN1995158,DBLP:conf/srds/MilosevicHS11}. Using BToB and an 
intermediary SBDLO the authors demonstrated how \atomicappends{} may be 
achieved in systems that suffer Byzantine failures. However, BToB is a strong 
primitive, and requires consensus to be solved. So one may ask: \textit{Is it 
possible to implement $\atomicappends{}$ without solving consensus?}

It was shown in~\cite{gupta2016non} that cryptocurrencies do not need consensus to be implemented. From a theoretical point of view, it was shown in~\cite{DBLP:conf/podc/GuerraouiKMPS19} that, assuming one process per account, the consensus number of cryptocurrencies is $1$. A non-sequential specification of money transfer was introduced in~\cite{DBLP:journals/eatcs/AuvolatFRT20}. It follows that Byzantine transactional systems do not necessarily need consensus, but rather can be implemented on top of less powerful data structures. In a similar manner, in this work, we observe
that intermediary S(B)DLOs and strong primitives like BToB 
\cite{DBLP:conf/srds/MilosevicHS11}, may not be 
necessary to allow interoperability between multiple DLOs. 
Note that the goal of the intermediate S(B)DLO is to collect the 
records to be appended atomically, so that when all the records involved are 
in the S(B)DLO, then the actual records are appended in their respective DLOs.
It is apparent that, for $\atomicappends$, the order of the records in the intermediary data structure is not important, but rather the membership 
property required redirects to a \textit{set} data structure.

A relevant distributed set data structure 
was presented by Shapiro et al. in~\cite{Shapiro11}
with the introduction of Conflict-Free Replicated Data Types (CRDTs). A CRDT is a data structure that can be replicated in multiple 
 network locations. CRDTs have the property that each replica can be updated independently and concurrently, but it is always mathematically possible to resolve any inconsistencies between any pair of replicas, leading eventually all the replicas to a consistent converged value when the communication  between the replica hosts is stabilized. 
 A \emph{Grow-Only Set} (G-Set) is such a CRDT, that supports operations \act{add} and \act{lookup} only. The \act{add} operation modifies the local state of the object by a union of
 the value of the set with the element we want to insert. Since \act{add} is based on union, and union is commutative, the \gset{} implementation converges. In \cite{Shapiro11} (and other subsequent works), implementations of \gset{}s where given 
 in a crash-prone environment. In order to utilise a \gset{} in more practical setups (like the ones considered in cryptocurrencies) 
 we need to examine whether such data structure is possible 
 when Byzantine failures are present in the system.
 

 Chai and Zhao \cite{DBLP:conf/IEEEscc/Chai014} have considered the implementation of CRDTs against Byzantine failures. In particular, they describe possible threats that clients and servers can either face or cause to CRDTs, and they show a possible solution to fulfil CRDT requirements in that failure model. Their solution relies on an external synchronization service for two main purposes: to guarantee linearizable {\em reads} and {\em writes}, and to prevent server partitions caused by Byzantine behaviour. As a consequence, multiple Byzantine failures or slow processes may lead their approach to essentially always run their ``state synchronization" mechanism letting the whole data structure rely on the synchronisation service. {For the implementation of the synchronisation service 
 they either utilize} a central entity, or solve consensus over a distributed set of nodes.\vspace{.5em}

\noindent{\bf Contributions.} In this work we examine whether \gset{}s can be implemented when Byzantine processes are assumed in the system, 
{without using consensus. We show that an implementation of an eventually consistent \cite{VEC2009} \gset is possible, and we demonstrate how such data structure can be used}
to solve \atomicappends{} and other related problems. In 
particular, our itemized contributions are the following:
\begin{itemize}
    \item Provide a formal definition of a Byzantine Grow-only Set Object (\BDGSO). [Section~\ref{sec:Spec}]
    \item {Provide an implementation for an eventually consistent  \BDGSO. We consider such a consistency model since, although it provides weaker guarantees {than other consistency models,  it is easier and more efficient to implement, while being 
    powerful enough to be used in the type of applications we consider (described next).}} [Section~\ref{GSetimplementation}]
    %
    \item Use \BDGSO{}s to implement:
    \begin{itemize}
    \item Consensus-free Byzantine \atomicappends. [Section~\ref{sec:AtomicAppends}]
    \item Consensus-free Byzantine {\em Atomic Adds}. This is the analogous problem of atomic appends where records must be added in an atomic way to different \BDGSO{}s.
     {This problem could be applicable in blockchain-like systems in which the ordering of the records is not important; what is important is that the records are added in the corresponding unordered blockchains (G-Sets). {An example could be a system of G-Sets that implement personal calendars, so the records in the sets are meetings. Then, fixing a two-person meeting would imply an Atomic Add of the meeting data in the calendar of both persons.}} [Section~\ref{sec:AtoomicAdds}]
    \item Consensus-free single-writer BDLOs. This data structure can be suitable to implement whatever system that requires total order among data produced by a single writer. A punch in/out system for a company is an example of such an application in which a single writer, the employee, appends records only to his/her own ledger of presences. A cryptocurrency can be another suitable application, with one BDLO per account, because of the need to order transactions in relation to money transfers issued by the only transaction signer. [Section~\ref{subsec:SW}]
    \end{itemize}
\end{itemize}

\section{The G-Set Object}
\label{sec:Spec}
In this section we provide the fundamental definition of a concurrent 
G-Set object. 

\subsection{Concurrent Objects and the G-Set Object} 
An \textit{object type} $T$ specifies $(i)$ the set of \emph{values} (or states) that any object $\obj{}$ of type $T$ can take, and
$(ii)$ the set of \textit{operations} that a process can use to modify or access the value of $\obj{}$.
An object $\obj{}$ of type $T$ is a \emph{concurrent object} if it is a shared object accessed by multiple processes~\cite{HW90, Raynal13}.
Each operation on an object $\obj{}$ consists of an \emph{invocation} event
and {its unique matching} \emph{response} event, that must occur in this order.
A \emph{history} of operations on $\obj{}$, denoted by $\hist{\obj{}}$, 
is the sequence of invocation and response events, starting with an invocation event. 
(The sequence order of a history reflects the real time ordering of the events.)
We say that a history $\hist{\obj{}}'$ \emph{extends} a history $\hist{\obj{}}$, if $\hist{\obj{}}$ is a prefix of $\hist{\obj{}}'$.

An operation $\op$ is \emph{complete} in a history $\hist{\obj{}}$, if $\hist{\obj{}}$ 
contains both the invocation and the matching response. 
A history $\hist{\obj{}}$ is {\em complete} if it contains only complete operations; otherwise it is {\em partial}~\cite{HW90, Raynal13}.
An operation 
$\op$ \emph{precedes} an operation $\op'$ (or $\op'$ \emph{succeeds} $\op$), denoted by $\op\bef\op'$, 
in $\hist{\obj{}}$, if the response event of $\op$ appears before the invocation event 
of $\op'$ in $\hist{\obj{}}$. Two operations are \emph{concurrent} if none precedes the other. 
A complete history $\hist{\obj{}}$ is \emph{sequential} 
if it contains no concurrent operations,
i.e., it is an alternative sequence of matching invocation and response events, starting with an invocation and ending with a response event.
A partial history is sequential, if removing its last event (that must be an invocation) makes it a complete sequential history.
%
A \emph{sequential specification} of an object $\obj{}$, 
describes the behavior of $\obj{}$ when accessed sequentially. In particular, the sequential specification of $\obj{}$ is the set of all possible sequential histories involving solely object $\obj{}$~\cite{Raynal13}.


A \emph{G-Set} $\ledger$ is a concurrent object that maintains a set $\ledger.S$ of \emph{records}
and supports two operations (available to any process $\pr$):
(i) $\ledger.\act{get}_\pr()$, and (ii) $\ledger.\act{add}_\pr(\rdr)$.
{A \emph{record} is any value drawn from an alphabet $A$.
\remove{A \emph{record} is a triple $\rdr=\tup{\tau, \pr, v}$, where $\tau$ is a {\em unique} record identifier from a set ${\mathcal T}$,
$\pr\in {\cal P}$ is the identifier of the process that created record $\rdr$, 
and $v$ is the data of the record drawn from an alphabet $A$. 
We will use $\rdr.p$ to denote the id of the process that created record $\rdr$; similarly, we define $\rdr.\tau$ and $\rdr.v$.}}
A process $\pr$ invokes a $\ledger.\act{get}_\pr()$ operation
to obtain the set $\ledger.S$ of records stored in the G-Set object $\ledger$
\footnote{We define only one operation to access the value of the G-Set for simplicity. In practice, other operations will also be available, 
{like {\em lookup($r$)} to check if a record $r$ is in $\ledger.S$.}}, and $\pr$ invokes a $\ledger.\act{add}_\pr(\rdr)$ operation to insert a new record $r$ in $\ledger.S$.
Initially, the set $\ledger.S$ is empty. Deleting or changing a record from $\ledger.S$ is not possible, as our objective is for the set to be immutable with respect to record modifications of any kind.

\begin{definition}
\label{def:sspec}
	The \emph{sequential specification} of a G-Set $\ledger$ over the sequential history $\hist{\ledger}$ is defined as follows. Let the initial value of 
	$\ledger.S=\emptyset$.
	If at the invocation event of an operation $\op$ in $\hist{\ledger}$ the value of the set $\ledger.S=V$, then:
	\begin{enumerate}
		\item  if $\op$ is a $\ledger.\act{get}_\pr()$ operation, then the response event of $\op$ returns $V$, and 
		\item if $\op$ is a $\ledger.\act{add}_\pr(\rdr)$ operation,
		then at the response event of $\op$, the value of the set in G-Set $\ledger$ is $\ledger.S=V\cup \{ r \}$.
	\end{enumerate}
\end{definition}

By comparing the sequential specification of a G-Set, as defined above, with the sequential specification of a Ledger Object  as defined in~\cite[Definition 1]{DLO_SIGACT18} (also see Appendix~\ref{appendix:DLO}), it follows that a Ledger is an {\em ordered} G-Set.


\subsection{Distributed G-Set Objects}
\label{sec:DL}
We now define distributed G-Set objects, \DGSO{} for short, and the class of eventually consistent \DGSO{}s. These definitions are general and do not rely on the properties of the underlying distributed system, nor on the type of failures that may occur.\newline 

\noindent A \emph{\bf distributed G-Set object} (\DGSO) is a concurrent
G-Set object
that is implemented in a distributed manner. In particular, a {\DGSO} is \textit{implemented} by a set of 
(possibly distinct and geographically dispersed) computing devices, that we refer as \emph{servers}. {Each server usually maintains a local copy (replica) of the {\DGSO}.}
We refer to the processes that invoke the $\act{get}$ and $\act{add}$ operations of the distributed G-Set as {\em clients}. 


\remove{In general, servers can fail. This leads to introducing mechanisms in the algorithm that implements
the DGSO to achieve fault tolerance, like replicating the G-Set. 
Additionally, the interaction of the clients with the servers
will have to take into account the faulty nature of individual servers,
as we discuss later in the section.}

Distribution and replication intend to ensure availability and survivability of the G-Set, in case a subset of the servers fails (by crashing or acting maliciously). 
At the same time, they raise the challenge of maintaining \emph{consistency} among the different views that different clients get of the \DGSO\footnote{This tradeoff is actually captured by the well-known CAP Theorem~\cite{brewer2012cap}.}.
\remove{:what is the latest value of the G-Set when multiple clients may send operation requests at different servers concurrently?}
Consistency semantics need to be in place to precisely describe the allowed values that a \act{get} operation may return when it is executed concurrently with other
\act{get} or \act{add} operations.

We now specify the safety properties of \DGSO\ with respect to
%
\textit{eventual consistency}~\cite{VEC2009}. {Essentially, these properties require that} if an
$\act{add}(r)$ operation completes, then \textit{eventually} all $\act{get}()$ operations return sets that contain record $r$. In a similar way, other consistency guarantees such as sequential, session, causal  and atomic consistencies could be formally defined. 

\begin{definition}
	\label{def:ec}
	A {\DGSO} $\ledger$ is {\em\bf eventually consistent} if, given any history $\hist{\ledger}$,
	\begin{enumerate}
	\item[(a)] Let $S$ be the set of records returned by any complete operation $\pi=\act{get}() \in \hist{\ledger}$. For each $r \in S$,
	there is an operation $\act{add}(r)$ whose {invocation event appears before the response event of $\pi$ in $\hist{\ledger}$, and}
		\item[(b)] for every complete operation $\ledger.\act{add}(r) \in \hist{\ledger}$, there exists a history $\hist{\ledger}'$ that extends 
		$H_{\ledger}$ such that,
		for every history $\hist{\ledger}''$ that extends $H'_{\ledger}$, every complete operation $\ledger.\act{get}()$ in $\hist{\ledger}'' \setminus \hist{\ledger}'$ returns 
		a set that contains $r$.  
	\end{enumerate}
\end{definition}

At this point, we would like to remark that, although eventual consistency provides weaker consistency guarantees when compared, for example, with linearizability~\cite{HW90}, {it is easier and more efficient to implement}, while 
it is powerful enough to be used in the type of applications that we  later consider (see Section~\ref{sec:applications}).

%
%
%



\subsection{Distributed Setting and Byzantine-tolerant \DGSO}

We consider a distributed setting consisting of processes (clients and servers) and an underlying communication graph in which each process can communicate with every other process. 

\noindent{\textbf{Asynchrony.}}
{Both processing and communication are asynchronous.} 
Therefore, each process proceeds at its own speed, which can
vary arbitrarily and remains always unknown to the other processes. 
Message transfer delays are  arbitrary but finite and remain
always unknown to the processes.
\vspace{.3em}

\noindent{\textbf{Failure Model.}}
No message is lost, duplicated or modified. 
  Processes (clients and servers) can fail arbitrarily, i.e., they can
  be Byzantine.  Specifically, we assume a \emph{Byzantine system} in
  which the number of servers that can arbitrarily fail is bounded by $f$, and
  in which the total number of servers, $n$, is at least $3f+1$. 
  For clients we assume that 
  any of them 
  can be Byzantine. 
  \vspace{.3em}

\noindent{\textbf{Public and private keys.}}
We assume that each process $p$ (client or server) has a pair of
public and private keys, and that the public keys have been distributed reliably to all the
processes that may interact with each other. Hence, we discard the possibility of spurious or fake processes (there cannot be Sybil attacks).
%
We also assume
that messages sent by any process (server or client) are
authenticated, so that messages corrupted or fabricated by Byzantine
processes are detected and discarded by correct
processes~\cite{CRISTIAN1995158}.  Communication channels between
correct processes are reliable but asynchronous. 
\vspace{.3em}

\noindent{\bf Byzantine-tolerant \DGSO{}s.} Our first aim is to propose an algorithm that implement an eventual-consistent \DGSO
$\ledger$ in a Byzantine asynchronous system.  {Here we present the
  properties that a \DGSO should satisfy with respect to \emph{correct
    processes}, given that Byzantine processes may return any
  arbitrary set or add any arbitrary record:}
%
\begin{itemize}
\item \emph{Byzantine Completeness (BC)}: All the $\act{get}()$ and $\act{add}()$ operations invoked by correct clients eventually
  complete.
    \item \emph{Byzantine Eventual Consistency (BEC)}: This is the property {of}
     Definition~\ref{def:ec} with respect to the $\act{get}()$ operations invoked by correct clients {and the $\act{add}(r)$ operations that insert the records $r$ returned in those $\act{get}()$ operations}. 
\end{itemize}
In the remainder, we say that a \DGSO is {\em Byzantine Tolerant},
denoted \BDGSO, and eventually consistent if it satisfies properties {BC} and BEC.\vspace{.3em}

\noindent{\bf Byzantine Reliable Broadcast.}
The algorithms presented in the next section to implement \BDGSO{}s are based on an underlying Byzantine Reliable Broadcast (BRB) service~\cite{abap,Raynal18}, which ensures that a message sent by a correct process is received by all correct processes, and that all correct processes {eventually} receive the same set of messages. The service provides two operations, BRB-broadcast and BRB-delivery; the first broadcasts a message to all processes, and the second delivers a message that was previously broadcast. The service is {
used by the servers, and from their
\remove{implemented by the servers, and from a client's}} point of view, the BRB service guarantees the following properties (as given in~\cite{Raynal18}):
\begin{itemize}
\item \textit{Validity}: if a correct process $p_i$ BRB-delivers a message $m$ from a correct process $p_j$, then $p_j$ BRB-broadcast $m$.
\item \textit{Integrity}: a message is BRB-delivered at most once by a correct server.
\item \textit{Termination 1 (local)}: if a correct process BRB-broadcasts a message, it BRB-delivers it.
\item \textit{Termination 2 (global)}: if a correct process BRB-delivers a message, all correct processes BRB-deliver it.
\end{itemize}

Validity relates outputs to inputs. Validity and integrity concern safety. 
Termination is on the fact that messages must be BRB-delivered; it concerns liveness. 
It follows (cf.~\cite{Raynal18}) that all correct processes BRB-deliver the same set of messages,
which includes all the messages they BRB-broadcast.


\section{Eventually Consistent \BDGSO\ Implementation}
\label{GSetimplementation}
In this section we provide the implementation of 
eventually consistent distributed G-Sets in an asynchronous distributed system with Byzantine failures. The implementation builds on a generic deterministic Byzantine-tolerant reliable broadcast service~\cite{abap,Raynal18}, which provides the properties given in the previous section. Our implementation is {\em optimally resilient}, in the sense that it can tolerate up to $f$ Byzantine servers, out of  $n\geq3f+1$ servers. 

\setlength{\columnsep}{10pt}
\begin{figure}[t]
		\begin{algorithm}[H]
			\caption{\small Client API and algorithm for Eventually Consistent Byzantine-tolerant Distributed G-Set Object $\ledger$}
			\label{code:ec-client-BGS}
			\begin{algorithmic}[1]
    				\State  \textbf{Init:} $c \leftarrow 0$
    				\Function{$\ledger.\act{get}$}{~} \Comment{{Invocation event}} \label{C3-L02}
    				\State $c \leftarrow c + 1$
    			        	\State
                             {\bf send} request {\sc get}($c$, $p$) to $3f +1$ different servers \label{C3-L04}
    				\State \textbf{wait} responses {\sc getResp}($c$, $i$, $S_i$) from $2f +1$ different servers \label{C3-L05}
    				\State $S \leftarrow \{r:$ record $r$ is in at least $f+1$ sets $S_i \}$ \label{C3-L06}
    				\State \textbf{return} $S$ \Comment{{Response event}} \label{C3-L07}
    				\EndFunction
    				\Function{$\ledger.\act{add}$}{$r$} \Comment{{Invocation event}} \label{C3-L08}
    				\State $c \leftarrow c + 1$
    				\State \textbf{send} request {\sc add}($c$, $p$, $r$) to $2f +1$ different servers \label{C3-L10}
    				\State \textbf{wait} responses {\sc addResp}($c$, $i$, {\sc ack}) from $f +1$ different servers \label{C3-L11}
    				\State \textbf{return} {\sc ack} \Comment{{Response event}} \label{C3-L12}
    				\EndFunction
    			\end{algorithmic}
		\end{algorithm}\vspace{-2em}
\end{figure}

\setlength{\columnsep}{10pt}
\begin{figure}[t]
	\begin{algorithm}[H]
			\caption{\small Server algorithm for Eventually Consistent Byzantine-tolerant Distributed G-Set Object}
			\label{code:ec-server-BGS}
			\begin{algorithmic}[1]
				\State \textbf{Init:} $S_i \leftarrow \emptyset$
				\Receive{{\sc get}($c$, $p$)} \Comment{{Signature of $p$ is validated}}
				\State \textbf{send} response {\sc getResp}($c$, $i$, $S_i$) to $p$ \label{C4-L03}
				\EndReceive
				
				\Receive{{\sc add}($c$, $p$, $r$)} \Comment{{Signature of $p$ is validated}}
				\If{($r \notin S_i$)} 
				\State \act{BRB-broadcast}({\sc propagate}($i$, {\sc add}($c$, $p$, $r$))) \label{C4-L06}
				\State \textbf{wait until} $r \in S_i$ \label{C4-L07}
				\EndIf
				\State  \textbf{send} response {\sc addResp}($c$, $i$, {\sc ack}) to $p$
				\EndReceive
				\Upon{\act{BRB-deliver}({\sc propagate}($j$, {\sc add}($c$, $p$, $r$)))} \Comment{{Signatures of $j$ and $p$ are validated}} \label{C4-L09}
				\If{($r \notin S_i$) and ({\sc add}($c$, $p$, $r$) was received from $f+1$ different servers $j$)} \label{C4-L10}
				\State $S_i \leftarrow S_i  \cup \{r\}$ \label{C4-L11}
				\EndIf
				\EndUpon
			\end{algorithmic}
		\end{algorithm}\vspace{-1em}
\end{figure}


Algorithm~\ref{code:ec-client-BGS} presents the specification of 
a client process, while Algorithm~\ref{code:ec-server-BGS} presents the specification of a server.
We now present a high level description of {how the two algorithms together implement an eventually consistent \BDGSO.}
%
\begin{itemize}

\item
When processing a $\ledger.\act{add}(r)$ operation a client
sends {\sc add} messages to a set of $2f+1$ servers, which guarantees that at least $f+1$ correct servers process it. 
{These correct servers broadcast the record $r$ to all servers using the BRB service, which leads to all correct servers $i$ adding $r$ to their replicas $S_i$ of the set.
When $f+1$ acknowledgement messages are received from the servers, the operation completes.}
{\remove{Since at least one of these servers is correct, this implies that, eventually, all correct servers will have $r$ in their sets.}}

\item
When processing a $\ledger.\act{get}()$ operation, a client requests their {replicas of the set} to $3f+1$ servers, via {\sc get} messages. 
We know that at least $2f+1$ of these servers will reply, since there are at least $2f+1$ correct servers among them. The first $2f+1$ responses contain at least $f+1$ responses from correct servers (and may contain up to $f$ responses from Byzantine servers). The $\ledger.\act{get}()$ operation returns the records that are contained in the {sets} of at least $f+1$ such responses, since each of these records is in at least one correct server set.

\item {Every server $i$ maintains a replica $S_i$ of the set $\ledger.S$. When server $i$ receives a {\sc get}($c$, $p$) message from a process $p$ it returns its current set $S_i$ to $p$. When $i$ receives a message {\sc add}($c$, $p$, $r$) from $p$, it makes sure $r$ has been included in its replica $S_i$ before sending an acknowledgment. {Server $i$ adds a record $r$ to its replica $S_i$ only if a corresponding add request has been processed} by at least one correct server. This is guaranteed by the BRB service and the requirement of receiving {\sc propagate}($j$, {\sc add}($c$, $p$, $r$)) from $f+1$ different servers. The properties of the BRB service also guarantee that once a record $r$ is delivered, then all correct servers will eventually add record $r$ to their replicas.}
\end{itemize}

We now provide the complete proof that the combination of Algorithms~\ref{code:ec-client-BGS}  and \ref{code:ec-server-BGS} implement an eventually consistent \BDGSO. In the proofs we consider that an operation $\pi$ is invoked in Lines \ref{C3-L02} or \ref{C3-L08} of Algorithm~\ref{code:ec-client-BGS}, and responds in Lines~\ref{C3-L07} or \ref{C3-L12} (resp.) of the same algorithm. 
{Let us first show that Byzantine Completeness holds, i.e., that all operations invoked by correct processes eventually complete.}

\begin{lemma}
\label{l-EC-BC}
Algorithms~\ref{code:ec-client-BGS} and~\ref{code:ec-server-BGS} guarantee Byzantine Completeness (BC) in a system in which at most $f$ out of $n \geq 3f+1$ servers are Byzantine.
\end{lemma}
\begin{proof}
Consider an operation $\ledger.\act{get}_p()$ invoked by a correct client $p$. We claim that the operation eventually completes. From Algorithm~\ref{code:ec-client-BGS}, Line~\ref{C3-L04}, $p$ sends a request {\sc get}($c$, $p$) to $3f+1$ servers and waits for
responses {\sc getResp}($c$, $i$, $S_i$) from $2f +1$ different servers. From the $3f+1$ servers to which the request is sent,
at most $f$ can be Byzantine, so at least $2f+1$ are correct servers that will eventually receive the {\sc get}($c$, $p$) message.
These servers will immediately send the corresponding response {\sc getResp}($c$, $i$, $S_i$) to $p$ (Line~\ref{C4-L03} of Algorithm~\ref{code:ec-server-BGS}).
When these responses are received eventually, the waiting in Line~\ref{C3-L05} of Algorithm~\ref{code:ec-client-BGS} will end. Since there is no other waiting condition, the operation will execute the return instruction and complete. 

Consider now an operation $\pi=\ledger.\act{add}_p(r)$ invoked by a correct client $p$. Then, the request {\sc add}($c$, $p$, $r$) is sent to $2f+1$ servers (Algorithm~\ref{code:ec-client-BGS}, Line~\ref{C3-L10}), and $p$ waits until responses {\sc addResp}($c$, $i$, {\sc ack}) are received from $f+1$ different servers. Since at most $f$ servers can be Byzantine, at least $f+1$ correct servers will receive and process the request. We prove that all these correct servers will send the corresponding response, the
waiting in Line~\ref{C3-L11} will end, and operation $\pi$ will complete.

Let us consider the set $C$ of correct servers that receive request {\sc add}($c$, $p$, $r$). Assume first that there is some server
$i \in C$ that has $r \in S_i$ when the request is received and processed. Then, server $i$ sends immediately response {\sc addResp}($c$, $i$, {\sc ack}) to $p$. Moreover, $r$ was inserted in $S_i$ in Line~\ref{C4-L11} of Algorithm~\ref{code:ec-server-BGS}, which implies that $i$ received via \act{BRB-deliver} at least $f+1$ messages {\sc propagate}() from different servers containing {\sc add}($c$, $p$, $r$) requests.
From the \emph{Termination 2} property of the BRB service, all correct processes will receive the same $f+1$ messages {\sc propagate}().
Consider any other correct server $j \in C$ that receives request {\sc add}($c$, $p$, $r$). If $r \in S_j$ when the request is received and processed, server $j$ sends the response {\sc addResp}($c$, $j$, {\sc ack}) to $p$ immediately. Otherwise, $r \notin S_j$ when the request is received and processed, and $j$ waits in Line~\ref{C4-L07}. From the above argument, eventually $r$ will be inserted in $S_j$, the waiting will end, and $j$ will send response {\sc addResp}($c$, $j$, {\sc ack}) to $p$.

Assume now that no correct server $i \in C$ has $r \in S_i$ when it receives request {\sc add}($c$, $p$, $r$). Then, all the (at least $f+1$) correct servers in $C$ that receive and process the request invoke \act{BRB-broadcast}({\sc propagate}($i$, {\sc add}($c$, $p$, $r$))) and start waiting in Line~\ref{C4-L07}. From the \emph{Termination 1} property of the BRB-service, if a correct server BRB-broadcasts a message, it also eventually BRB-delivers it. Moreover,
from \emph{Termination 2}, if it BRB-delivers the message, all correct servers also BRB-deliver it. So each correct server $i \in C$ will process in Lines~\ref{C4-L09}-\ref{C4-L11} messages {\sc propagate}($j$, {\sc add}($c$, $p$, $r$)) from at least $f+1$ different servers $j$. Hence, server $i$ will insert $r$ in $S_i$ in Line~\ref{C4-L11},
the waiting will end, and $i$ will send response {\sc addResp}($c$, $i$, {\sc ack}) to $p$.
\end{proof}

\begin{theorem}
\label{proof:ec}
Algorithms~\ref{code:ec-client-BGS} and~\ref{code:ec-server-BGS} implement an Eventually Consistent {\em \BDGSO}, in a system in which at most $f$ out of $n \geq 3f+1$ servers are Byzantine. 
\end{theorem}

\begin{proof}
We need to prove that Algorithms~\ref{code:ec-client-BGS} and~\ref{code:ec-server-BGS} guarantee Byzantine Completeness (BC) and Byzantine Eventual Consistency (BEC). BC is shown to be satisfied in Lemma~\ref{l-EC-BC}. Regarding Byzantine Eventual Consistency, we need to demonstrate properties (a) and (b) of Definition \ref{def:ec} with respect to the $\act{get}()$ operations invoked by correct clients and the $\act{add}(r)$ operations that insert the records $r$ returned in those $\act{get}()$ operations. Let $H_{\ledger}$ be any history including only invocation and response events of these operations.
		
{\em Property (a):} 
Consider a complete operation $\pi=\act{get}_p() \in \hist{\ledger}$ invoked by a correct client $p$, let $S$ be
the set returned by $\pi$, and consider any $r \in S$. From Line~\ref{C3-L06} of Algorithm~\ref{code:ec-client-BGS}, $r$ belongs to at least $f+1$ sets $S_i$ 
received in responses {\sc getResp}($c$, $i$, $S_i$) from a set $C$ of different servers.
All these responses must have been sent before the response event of $\pi$ (Line~\ref{C3-L07} of Algorithm~\ref{code:ec-client-BGS}).

Observe that $C$ contains at least one correct server $i$.
This mean that some correct server $i$ had $r \in S_i$ when it sent
the response {\sc getResp}($c$, $i$, $S_i$). A server $i$ only adds a record to its local set $S_i$ if that record was BRB-delivered in {\sc propagate}($j$, {\sc add}($c'$, $p'$, $r$))) 
from $f+1$ different servers $j$ (Line~\ref{C4-L10} of Algorithm~\ref{code:ec-server-BGS}). From the Validity property of the BRB service, this means that at least $f+1$ servers called
\act{BRB-broadcast}({\sc propagate}($j$, {\sc add}($c'$, $p'$, $r$)) in Line~\ref{C4-L06}. Again, since at least one of them is correct, at least one invocation of BRB-broadcast was done
by a process because it previously received a request {\sc add}($c'$, $p'$, $r$) from client $p'$. Hence the invocation of $\act{add}(r)$ must have preceded the reception of this request, and by transitivity must have preceded the response event of $\pi$. 

{\em Property (b):} This property holds if, for every complete operation $\ledger.\act{add}(r) \in \hist{\ledger}$,
there exists a time $t$ after which every $\ledger.\act{get}()$ operation invoked after $t$ returns sets $S$ that contains $r$.
Let us first consider a complete operation $\pi=\ledger.\act{add}_p(r) \in \hist{\ledger}$ invoked by a client $p$.
We claim that there is some correct server $i$ that eventually adds record $r$ to its replica $S_i$. This is true when $p$ is
Byzantine, since that is the requirement for an $\act{add}(r)$ operation of a Byzantine client to be considered. 

On the other hand, if $p$ is correct, let us assume for contradiction that no correct server $i$ adds record $r$ to its replica $S_i$.
Process $p$ sends request {\sc add}($c$, $p$, $r$) to $2f+1$ servers, out which at least $f+1$ are correct.
By assumption, $r \notin S_j$ when each of these servers $j$ processes the request, and hence all of them
execute \act{BRB-broadcast}({\sc propagate}($j$, {\sc add}($c$, $p$, $r$))) (Line~\ref{C4-L06} of 
Algorithm~\ref{code:ec-server-BGS}). Then, from the \emph{Termination 1} and \emph{Termination 2} properties of the
BRB service, some correct server $i$ will BRB-deliver at least $f+1$ messages {\sc propagate}($j$, {\sc add}($c$, $p$, $r$)) from different servers $j$, and then record $r$ will be added to $S_i$ in Line~\ref{C4-L11}. This is a contradiction, and some correct server
$i$ eventually adds record $r$ to its replica $S_i$ when client $p$ is correct.

Hence, we have that, independently of whether $p$ is correct, some correct server
$i$ added record $r$ to its set $S_i$.
Observe that a correct process $i$ only adds records to its replica $S_i$, in Line~\ref{C4-L11}, when BRB-deliver at least $f+1$ messages {\sc propagate}($j$, {\sc add}($c$, $p$, $r$)) from different servers $j$. Then, if $i$ adds $r$ to $S_i$, from the \emph{Termination 2} property 
all correct servers will eventually BRB-deliver at least $f+1$ messages {\sc propagate}($j$, {\sc add}($c$, $p$, $r$)) from different servers $j$, and they will all add $r$ to their replicas.

Let $t$ be the first time all correct servers have $r$ in their corresponding replica. Then, for every $\ledger.\act{get}()$ operation invoked after $t$, the responses from correct servers collected in Line~\ref{C3-L05} of Algorithm~\ref{code:ec-client-BGS} 
have replicas $S_i$ with record $r$. Since there at least $f+1$ responses from correct servers, in Line~\ref{C3-L06} $r$ is included in the set $S$, which is then returned by $\ledger.\act{get}()$.
\end{proof}
%
%
%
%
\section{Applications of \BDGSO{}s}
\label{sec:applications}
In this section we demonstrate the usability of \BDGSO{}s by using them to provide consensus-free solutions to the Atomic Appends and Atomic Adds problems, as well as {a consensus-free construction of a Single-Writer Byzantine-tolerant Distributed Ledger Object (BDLO).}

\subsection{The Atomic Appends Problem}
\label{sec:AtomicAppends}
The {\em Atomic Appends} problem was introduced in~\cite{DLO_SIGACT18} as a basic interconnection problem among distributed ledgers (DLOs); see Appendix~\ref{appendix:DLO} for basic definitions with respect to DLOs. Informally, {Atomic Appends} requires that several records must be appended in their corresponding DLOs, so that either
\emph{all} records are appended (each in the appropriate DLO)
or
\emph{none} is appended to any DLO. In~\cite{BAA2020}, the problem was formulated (and solved) in the presence of Byzantine servers and clients.\vspace{.3em}

\noindent{\bf Definition of the problem.}
For completeness, we provide the formal definition as given in~\cite{BAA2020}.
A record $r$ \emph{depends} on a record $r'$ if $r$ may
be appended on its intended BDLO, say $\mathcal{L}$, only if $r'$ is
appended on its intended BDLO, say $\mathcal{L}'$. Two records, $r$ and $r'$ are
\emph{mutually dependent} if $r$ depends on $r'$ and $r'$ depends on
$r$.
\begin{definition}[$2$-AtomicAppends~\cite{BAA2020}]
	\label{def:2AA}
	Consider two clients, $p$ and $q$, with mutually dependent
        records $r_p$ and $r_q$.  We say that records $r_p$ and $r_q$
        are {\em appended atomically} in BDLO $\mathcal{L}_p$ and BDLO
        $\mathcal{L}_q$, respectively, when:
\begin{itemize}
 \item AA-safety (AAS): \emph{The record $r_p$ of a
                  correct client $p$ is appended in $\mathcal{L}_p$ only
                  if the record of the other client $q$ (which may be
                  correct or not) is also appended in $\mathcal{L}_q$.}
 \item AA-liveness (AAL): \emph{If both $p$ and $q$
                  are correct, then both records are appended
                  eventually.}
\end{itemize}
\end{definition}

Observe that it is not possible to prevent a faulty client $q$
  from appending its record $r_q$, even if the correct client $p$ does
  not append its record. What the safety property AAS guarantees is that the opposite
  cannot happen. This is analogous of the property in atomic
  cross-chain swaps~\cite{DBLP:conf/podc/Herlihy18} that a correct
  process cannot end up worse than at the beginning.

We say that an algorithm {\em solves} the $2$-AtomicAppends problem\footnote{The $k$-\emph{AtomicAppends} problem, for $k\ge 2$, is a
generalization of the $2$-AtomicAppends that can be defined in the
natural way: $k$ clients, with $k$ mutually dependent records, to be
appended to $k$ BDLOs. To keep the presentation simple, we focus in the case of $k=2$.} 
under a given {system, if it guarantees 
  properties AAS and AAL of Definition~\ref{def:2AA} in every
  execution.} Since we consider Byzantine failures, our
system model with respect to the \atomicappends{} problem is such that
the correct processes want to proceed with the append of the records
{(to guarantee liveness AAL),} while the Byzantine processes may try
to get correct clients to append without the Byzantine clients doing so (to prevent safety AAS).\vspace{.3em}

\noindent{\bf Prior solution.} The solution of $2$-AtomicAppends in~\cite{BAA2020}, following the work in~\cite{DLO_SIGACT18}, uses an auxiliary, special purpose BDLO, called Smart BDLO (SBDLO) to aggregate and coordinate the append of
multiple records. In a nutshell, the solution in~\cite{BAA2020} is as follows. 
Consider two clients, $p$ and $q$, that wish to append atomically two mutually dependent records, $r_p$ and $r_q$, in BDLOs $\mathcal{L}_p$ and
$\mathcal{L}_q$, respectively. Then, they both send matching {\em atomic append requests} to the SBDLO. Once both requests are received by the SBDLO (otherwise the atomic append never takes place), the servers implementing the SBDLO proceed to append each record to the appropriate BDLOs. In particular, the servers of the SBDLO now become clients issuing the corresponding appends to the servers implementing the DBLOs $\mathcal{L}_p$ and $\mathcal{L}_q$ (each BDLO could be implemented by different servers, as these are essentially different distributed ledger systems). The whole process involves several algorithms: the algorithm run by the clients to issue the atomic append request, the algorithm run by servers to implement the SBDLO, and the algorithm run by the servers of the SBDLO (as clients) with the servers of each individual BDLO. Once both append operations are completed, the SBDLO servers acknowledge this to clients $p$ and $q$. It is shown that the combination of these algorithms guarantee Properties AAS and AAL above, despite having Byzantine servers and clients.\vspace{.3em} 

\noindent{\bf Our approach.} In this work we treat the part of the individual BDLOs ($\mathcal{L}_p$ and $\mathcal{L}_q$) implementations as black boxes and we focus on the auxiliary entity that is used for coordinating the atomic append requests. In~\cite{BAA2020}, the SBDLO, being a Distributed Ledger object, required the use of a Byzantine Total-order Broadcast~\cite{DBLP:conf/srds/MilosevicHS11} service. It was shown in~\cite{DLO_SIGACT18} that consensus is required for implementing a (B)DLO; this is because of the strong prefix property of (B)DLOs (see Appendix~\ref{appendix:DLO}), which requires that records must be totally ordered. Hence, atomic appends was solved using consensus to implement the SBDLO. However, one can notice that in the auxiliary entity, the atomic append requests do not need to be totally ordered. It is sufficient to only keep track whether both requests have been made. In other words, \emph{why keeping these requests in a sequence, and not in a set?}

In this respect, we show that instead of using a special purpose BDLO as the auxiliary entity, we can simply use a special purpose eventually consistent \BDGSO, which we will be referring as S\BDGSO. As we have seen in Section~\ref{GSetimplementation}, eventually consistent \BDGSO{}s can be implemented without consensus (instead of a Byzantine total-order broadcast service, we use only a Byzantine reliable broadcast service), yielding a {\em consensus-free solution to Atomic Appends} (with respect to the actual atomic append requests).\vspace{.3em}


    	\begin{algorithm}[t!]
			\caption{\small API  for the $2$-\act{AtomicAppend} of records $r_p$ and $r_q$ in ledgers $\mathcal{L}_p$ and $\mathcal{L}_q$ by clients $p$ and $q$, respectively, using S\BDGSO{} $\ledger$. Code for Client $p$.} 
			\label{code:sdlo-client-1}
		        \begin{algorithmic}[1]
    				\Function{$\act{AtomicAppends}$}{$p, \{p,q\}, r_p, \mathcal{L}_p, r_q$}
    				\State $\ledger.{\sc add}(\tup{p, \{p,q\}, r_p, \mathcal{L}_p, r_q})$
    				\State \textbf{return} {\sc ack}
    				\EndFunction
    				\State {// Client $p$ will know the \atomicappends{} operation was completed successfully
    				when it receives notifications from $f+1$ different S\BDGSO{} servers. // }
    			\end{algorithmic}
		\end{algorithm}

\noindent{\bf Our solution.} Algorithm~\ref{code:sdlo-client-1} specifies how processes $p$ and $q$ delegate the task of appending their records in the respective ledgers. They do so by adding in the S\BDGSO{} a
description of the \atomicappends{} operation to be completed. Client $p$ uses the {\sc ADD}
operation to provide the S\BDGSO{} with the data it requires to complete
the \atomicappends{}, namely the participants in the \atomicappends{},
the record $r_p$, the BDLO $\mathcal{L}_p$, and the record $r_q$ the other
client is appending. (The other client must do the same.)

\setlength{\columnsep}{10pt}
\begin{figure}[t!]
	\begin{algorithm}[H]
			\caption{\small Smart  Byzantine-tolerant  \DGSO; Only the code for the {\sc ADD} operation is shown; Code for Server $i$}
			\label{code:SBGS}
			\begin{algorithmic}[1]
				\State \textbf{Init:} $S_i \leftarrow \emptyset$
				\Receive{{\sc add}($c$, $p$, $r$)} \Comment{{Signature of $p$ is validated}}
				\If{($r \notin S_i$)} 
				\State \act{BRB-broadcast}({\sc propagate}($i$, {\sc add}($c$, $p$, $r$)))
				\State \textbf{wait until} $r \in S_i$ 
				\EndIf
				\State  \textbf{send} response {\sc addResp}($c$, $i$, {\sc ack}) to $p$
				\EndReceive
				\Upon{\act{BRB-deliver}({\sc propagate}($j$, {\sc add}($c$, $p$, $r$)))} \Comment{{Signatures of $j$ and $p$ are validated}} 
				\If{($r \notin S_i$) and ({\sc add}($c$, $p$, $r$) was received from $f+1$ different servers $j$)} 
				\State $S_i \leftarrow S_i  \cup \{r\}$ 
				\If {($r.v=\tup{p, \{p,q\}, r_p, \mathcal{L}_p, r_q}$) and \label{l:record1}\\
 \hskip\algorithmicindent\hspace{0.5cm} ($\exists r' \in S_i : r'.v=\tup{q, \{p,q\}, r_q, \mathcal{L}_q, r_p}$)} \label{l:record2}
    				\State $\mathcal{L}_p.{\append}(r_p)$; \label{l:calls} 
    				 $\mathcal{L}_q.{\append}(r_q)$			
    				\State {Notify clients $p$ and $q$
    				that records $r_p$ and $r_q$ have been appended to $\mathcal{L}_p$ and $\mathcal{L}_q$} \label{l:end}
				\EndIf
				\EndIf
				\EndUpon
			\end{algorithmic}
		\end{algorithm}\vspace{-2em}
\end{figure}

For the S\BDGSO{}, it suffices to implement an eventually consistent \BDGSO{} in which up to 
$f$ servers out of $n\ge 3f+1$
are Byzantine, but that \emph{only allows the creator of a record to add it} (signatures are used for this purpose).
Algorithm \ref{code:SBGS} describes the {\sc add} operation of the S\BDGSO{}
(the rest of the algorithm is as in
Algorithm~\ref{code:ec-server-BGS}).  As expected, it is very similar to the
implementation of a \BDGSO{}, but with an important difference: 
Every time a record $r$ is added to the sequence $S_i$, it is checked
whether a matching record $r'$ is already there. This is the case if
$r.v=\tup{p, \{p,q\}, r_p, \mathcal{L}_p, r_q}$, and $r'.v=\tup{q,
  \{p,q\}, r_q, \mathcal{L}_q, r_p}$. If so, the corresponding append
operations are issued in the respective BDLOs $\mathcal{L}_p$ and
$\mathcal{L}_q$ (the implementation of this part is the one described in~\cite{BAA2020}). {So, essentially the servers implementing the S\BDGSO,
  become proxies of clients $p$ and $q$, and once the above condition
  is met, they issue the corresponding appends. When these appends are
  successful, the servers implementing the ledgers $\mathcal{L}_p$ and
  $\mathcal{L}_q$, acknowledge the S\BDGSO{} servers. In turn, the S\BDGSO{}
  servers notify clients $p$ and $q$ that records $r_p$ and $r_q$ have
  been appended to $\mathcal{L}_p$ and $\mathcal{L}_q$, respectively.}
{Clients $p$ and $q$ will know that the \atomicappends{} operations
  was completed successfully when they receive these notifications
  from at least $f+1$ different S\BDGSO\ servers.}

\remove{As mentioned above, each of the ledgers $\ledger_p$ and $\ledger_q$
are BDLOs with a known, {bounded} set $N$, of at least $2t+1$ clients
({which are the servers implementing the SBDLO $\ledger$}), out of
which at most $t$ can be Byzantine.  These ledgers are implemented in
a system of at least $2f+1$ servers out of which at most $f$ can be
Byzantine, as presented in Algorithm~b-ByDL (Code
\ref{impl:at-server-bounded}). Hence, a record is appended only if at
least $t+1$ clients from $N$ issue append operations of the
record. {Notice that, differently from the case of ad-hoc clients,
  in the case
  of SBDLO at least $t+1$ correct SBDLO servers will receive the
  requests by the external clients $p$ and $q$ and will issue the same
  \append{} operation in ledgers $\ledger_p$ and $\ledger_q$, making
  bounded BDLOs a practical system.} Moreover, Line 2 of Code~\ref{impl:at-server-bounded} is modified to verify that a client $p$
attempting to append is in fact in the set $N$ of authorized clients. Figure~\ref{fig:atomicappends} illustrates how an \emph{AtomicAppends} procedure works

\begin{figure}[!t]
\centering
\includegraphics[width=0.9\linewidth]{AtomicAppends3.pdf}
 \caption{This drawing visualizes how  the different algorithms involved in an atomic append interact.}
 \label{fig:atomicappends}
\end{figure}
}

\begin{theorem}
\label{SBDGO-correct}
The combination of Algorithm~\ref{code:sdlo-client-1} and Algorithm~\ref{code:SBGS} solves the $2$-AtomicAppends problem.
\end{theorem}
\noindent
The proof follows from the one in~\cite{BAA2020}, taking into consideration the above discussion.
\vspace{.3em}


\noindent{\bf\em Remark:} Following the approach described in~\cite[Section IV-B]{BAA2020},
the S\BDGSO{}  can be replaced by a  “classical”  \BDGSO{} $\ledger$ and  the  use  of  a  set of “helper” processes. The helper processes take upon themselves the task of consulting $\ledger$ periodically in order to find new matching descriptions of and \atomicappends{} operation. When such a match is found, they complete the corresponding appends (as done in Lines~\ref{l:record2}-\ref{l:end} of Algorithm~\ref{code:SBGS}).

\subsection{The Atomic Adds Problem}
\label{sec:AtoomicAdds}
 Inspired by the Atomic Appends problem, one could define the analogous problem on \BDGSO{}s, {\em Atomic Adds}:  several records must be added in their corresponding \BDGSO{}s, and either all records are added (each in the appropriate \BDGSO) or none is added. 
 The formal definition follows that of the Atomic Appends.
 
 \begin{definition}[$2$-AtomicAdds]
	\label{def:2AAdds}
	Consider two clients, $p$ and $q$, with mutually dependent
        records\footnote{The definition of mutually dependent records is as in the case of Atomic Appends, but for \BDGSO{}s instead of BDLOs.} $r_p$ and $r_q$.  We say that records $r_p$ and $r_q$
        are {\em added atomically} in \BDGSO{} $\ledger_p$ and \BDGSO{}
        $\ledger_q$, respectively, when:
\begin{itemize}
 \item AAd-safety (AAdS): \emph{The record $r_p$ of a
                  correct client $p$ is added in $\ledger_p$ only
                  if the record of the other client $q$ (which may be
                  correct or not) is also added in $\ledger_q$.}
 \item AAd-liveness (AAdL): \emph{If both $p$ and $q$
                  are correct, then both records are added
                  eventually.}
\end{itemize}
\end{definition}

The $k$-AtomicAdds problem can be defined in the natural way: $k$ clients, with $k$ mutually dependent records, to be
appended to $k$ \BDGSO{}s.
It is not difficult to see that a consensus-free algorithmic solution for this problem can be derived by simple modifications of our solution to the Atomic Appends problem and the use of the \BDGSO\ implementation of Section~\ref{GSetimplementation}.\vspace{.3em} 

\noindent{\bf Atomic Adds API and server code.} The Atomic Adds API, shown in Algorithm~\ref{addsAPI}, is very close to Algorithm~\ref{code:sdlo-client-1}. The main difference is the content of the data to be added (since now we have G-Sets and not ledgers).
	\begin{algorithm}[t]
			\caption{\small API  for the $2$-\act{AtomicAdds} of records $r_p$ and $r_q$ in \BDGSO{}s $\ledger_p$ and $\ledger_q$ by clients $p$ and $q$, respectively, using S\BDGSO\ $\ledger$. Algorithm for Client $p$.} 
			\label{addsAPI}
		        \begin{algorithmic}[1]
    				\Function{$\act{AtomicAdds}$}{$p, \{p,q\}, r_p, \ledger_p, r_q$}
    				\State $\ledger.{\sc add}(\tup{p, \{p,q\}, r_p, \ledger_p, r_q})$
    				\State \textbf{return} {\sc ack}
    				\EndFunction
    				\State {// Client $p$ will know the Atomic Adds operation was completed successfully
    				when it receives notifications from $f+1$ different S\BDGSO{} servers. // }
    			\end{algorithmic}
		\end{algorithm}
The code run by the servers of S\BDGSO\ is the same as in Algorithm~\ref{code:SBGS}, with the difference that Lines \ref{l:record1} and \ref{l:record2} check for matching atomic add requests, and once found, in Line~\ref{l:calls} will call the corresponding add operations, 
$\ledger.{\sc add}(r_p)$ and $\ledger.{\sc add}(r_q)$, which are implemented by the algorithms in Section~\ref{GSetimplementation}.
{Note that the condition in Line~\ref{C3-L10} of Algorithm~\ref{code:ec-server-BGS} may have to be expanded in order to prevent the 
(up to $f$) Byzantine servers that implement the S\BDGSO from adding spurious records in $\ledger_p$ and $\ledger_q$. This may be achieved
adding a record $r$ in these \DGSO{}s only if at least $f+1$ clients (the servers of the S\BDGSO) request it to be added, similarly as done
in~\cite{BAA2020}. 
}

The sequence of events is now as described in the Atomic Appends solution, with the difference that no BDLOs are now involved, only \BDGSO{}s.
Putting everything together, we obtain the following, whose proof details are omitted (it is essentially a restatement of the corresponding observations in the atomic appends proof in~\cite{BAA2020}, and the correctness of the algorithms in Section~\ref{GSetimplementation}):

\begin{theorem}
\label{SBDGOAdds-correct}
The combination of the API of Algorithm~\ref{addsAPI} with the revised version of Algorithm~\ref{code:SBGS}, and Algorithms \ref{code:ec-client-BGS} and \ref{code:ec-server-BGS} yield a solution to the $2$-AtomicAdds problem.
\end{theorem}

As noted above, the S\BDGSO{}  could be replaced by a  “classical”  \BDGSO{}  and  the  use  of  a  set of “helper” processes. See~\cite[Section IV-B]{BAA2020} for this approach.

\subsection{Consensus-free Single-Writer BDLO}
\label{subsec:SW}
\setlength{\columnsep}{10pt}
\begin{figure}[t]
		\begin{algorithm}[H]
			\caption{\small Client API and algorithms for Eventually Consistent Single-Writer BDLO $\mathcal{L}$ with $n \geq 4f+1$ and writer process $w$}
			\label{code:ec-client-SWBDLO}
			\begin{algorithmic}[1]
    				\State  \textbf{Init:} $c \leftarrow 0$, $k \leftarrow 0$
    				\Function{$\mathcal{L}.\act{get}$}{~}
    				\State $c \leftarrow c + 1$
    			        	\State
                             {\bf send} request {\sc get}($c$, $p$) to $3f +1$ different servers
    				\State \textbf{wait} responses {\sc getResp}($c$, $i$, $S_i$) from $2f +1$ different servers \label{l:dlo-wait}
    				\State $A \leftarrow \{r:$ record $r$ is in at least $f+1$ sets $S_i \}$\label{code:swbdlo-client-get-f1check}
    				\State $S \leftarrow \{r \in A : (r.k=1) \lor (\exists r' \in A: r.k=r'.k+1) \}$ \label{code:swbdlo-client-get-lcp}
    				\State \textbf{return} sequence $(\rho_1||\rho_2||\ldots||\rho_m)$, where $m=|S|$ and $r_i=(i,\rho_i) \in S$ \label{code:swbdlo-client-seq-create}
    				\EndFunction
    				\Function{$\mathcal{L}.\act{append}$}{$\rho$} \Comment{Can only be called by process $w$}
    				\State $c \leftarrow c + 1$, $k \leftarrow k + 1$ \label{code:swbdlo-client-indexesIncrement}
    				\State $r \leftarrow (k, \rho)$
    				\State \textbf{send} request {\sc add}($c$, $w$, $r$) to $\lfloor n/2 \rfloor +2f+1$ different servers \label{code:swbdlo-client-add}
    				\State \textbf{wait} responses {\sc addResp}($c$, $i$, {\sc ack}) from $f +1$ different servers \label{code:swbdlo-client-addWait}
    				\State \textbf{return} {\sc ack}
    				\EndFunction
    			\end{algorithmic}
		\end{algorithm}
		
		\begin{algorithm}[H]
			\caption{\small Server algorithm for Eventually Consistent Single-Writer BDLO $\mathcal{L}$ with $n \geq 4f+1$ and writer process $w$}
			\label{code:ec-server-SWBDLO}
			\begin{algorithmic}[1]
				\State \textbf{Init:} $S_i \leftarrow \emptyset$, $T \leftarrow \emptyset$ 
				\Receive{{\sc get}($c$, $p$)} 
				\State \textbf{send} response {\sc getResp}($c$, $i$, $S_i$) to $p$ \label{code:swbdlo-server-getResp}
				\EndReceive
				
				\Receivex{{\sc add}($c$, $w$, $r$)}{$w$} \label{code:swbdlo-server-addReceive}
				\If{($r.k \notin T$)}  \label{code:swbdlo-server-ifnotinT}
				\State \act{BRB-broadcast}({\sc propagate}($i$, {\sc add}($c$, $w$, $r$))) \label{code:swbdlo-server-broadcast}
				\State $T \leftarrow T \cup \{ r.k \}$ \label{code:swbdlo-server-indexUpdate}
				\State \textbf{wait until} $r \in S_i$ \label{code:swbdlo-server-recordAddWait}
				\EndIf
				\State  \textbf{send} response {\sc addResp}($c$, $i$, {\sc ack}) to $w$
				\EndReceivex
				\Upon{\act{BRB-deliver}({\sc propagate}($j$, {\sc add}($c$, $w$, $r$)))}
				\If{({\sc add}($c$, $w$, $r$) was received from $\lfloor n/2 \rfloor +f+1$ different servers $j$)} \label{code:swbdlo-server-delivery}
				\State $S_i \leftarrow S_i \cup \{ r \}$ \label{code:swbdlo-server-setUpdate}
				\EndIf
				\EndUpon
			\end{algorithmic}
		\end{algorithm}\vspace{-1em}
\end{figure}

The \BDGSO{} can also be used to implement a Single-Writer BDLO without relying on consensus. This is obtained with a 
\BDGSO{} that allows only a single writer process $w$ to add records, in which each record has an index determining its position in the BDLO sequence, and that does not allow adding more than one record with the same index.
Allowing only add operations from $w$ is trivially achieved by validating the signature when a request is received by a server,
and will not be done explicitly in our algorithms.
To prove correctness we need to show that any execution of the Single-Writer BDLO $\mathcal{L}$ 
we implement satisfies the Byzantine Completeness and Byzantine Eventual Consistency properties, but redefined for the
$\mathcal{L}.\act{append}()$ and $\mathcal{L}.\act{get}()$ operations, and sequences instead of sets (see Appendix~\ref{appendix:DLO}).
Additionally, 
the Byzantine Strong Prefix property, as defined in \cite{BAA2020}, must be satisfied as well.

\begin{definition}[Byzantine Strong Prefix~\cite{BAA2020}]
    \label{def:bsp}
    If two \textit{correct clients} of a BDLO $\mathcal{L}$ issue two $\mathcal{L}.\act{get}()$ operations that return record sequences $S$ and $S'$ respectively, then either $S$ is a prefix of $S'$ or vice-versa.
\end{definition}

Algorithm~\ref{code:ec-client-SWBDLO} presents the API and the code executed by a client of the Single-Writer BDLO $\mathcal{L}$, while 
Algorithm~\ref{code:ec-server-SWBDLO} presents the code executed by the servers that implement it. 
These algorithms require that the number of servers $n$ satisfies $n \geq 4f+1$.
As can be seen, the append operation assigns an index $k$ to every record data $d$ appended by $w$, so the record added is in fact the pair $r=(k, d)$. Observe that
Algorithms~\ref{code:ec-client-SWBDLO} and \ref{code:ec-server-SWBDLO} are very similar to Algorithms~\ref{code:ec-client-BGS} and~\ref{code:ec-server-BGS}, but have a few differences. (1) In Algorithm~\ref{code:ec-client-SWBDLO}, $\mathcal{L}.\act{append}(d)$
adds an index $k$ to each record and sends the append requests to a potentially much larger set of $\lfloor n/2 \rfloor +2f+1$ servers, 
while $\mathcal{L}.\act{get}()$ filters the set to be returned so it is a sequence of records with consecutive indices. On its hand,
(2) Algorithm~\ref{code:ec-server-SWBDLO} avoids appending different records with the same index $r.k$ by using this field for comparisons, keeping track in $T$ of the indices that have been BRB broadcast, and collecting at least $\lfloor n/2 \rfloor +f+1$ messages
{\sc propagate}($j$, {\sc add}($c$, $w$, $r$)) before adding $r$ to the set.
Observe that the requirement on $n$ comes from the fact that the append requests are sent to $\lfloor n/2 \rfloor +2f+1$ servers (and hence
$n \geq \lfloor n/2 \rfloor +2f+1$).

\begin{theorem}
Algorithms \ref{code:ec-client-SWBDLO} and \ref{code:ec-server-SWBDLO} implement an eventually consistent Single-Writer BDLO~$\mathcal{L}$.
\end{theorem}
\begin{proof} We will first show Byzantine Completeness, then Byzantine Eventual Consistency and lastly Byzantine Strong Prefix.

    \noindent{\bf\em Byzantine Completeness}: Let us consider an $\mathcal{L}.\act{get}()$ operation invoked by a correct client $p$. Then
    request {\sc get}($c$, $p$) is sent to $3f+1$ different servers so at least $2f+1$ correct ones will eventually send back their responses; in fact correct servers simply answer back in Line~\ref{code:swbdlo-server-getResp} of Algorithm~\ref{code:ec-server-SWBDLO} with a {\sc getResp}($c$, $i$, $S_i$) containing their local $S_i$. Then, the condition of the wait operation in Line~\ref{l:dlo-wait} is eventually
    satisfied and the operation completes.
    
    Let us now assume that $w$ is correct, and consider an $\mathcal{L}.\act{append}()$ operation. Then, requests {\sc add}($c$, $w$, $r$) will be sent (Line~\ref{code:swbdlo-client-add} of Algorithm~\ref{code:ec-client-SWBDLO}) to $\floor{n/2}+2f+1$ servers, so at least $\floor{n/2}+f+1$ correct ones will receive it. Since $w$ is correct, it increments $k$ before sending the {\sc add}($c$, $w$, $r$) messages (Line~\ref{code:swbdlo-client-indexesIncrement} of Algorithm~\ref{code:ec-client-SWBDLO}), so the same index $k$ is not
    used twice. Then, every correct process that receives {\sc add}($c$, $w$, $r$) finds that $r.k \notin T$ (since $T$ is updated in Line~\ref{code:swbdlo-server-indexUpdate} of Algorithm~\ref{code:ec-server-SWBDLO} only after this check).
    Hence, the \act{BRB-broadcast}({\sc propagate}($i$, {\sc add}($c$, $w$, $r$))) in
    Line~\ref{code:swbdlo-server-broadcast} is called at least by $\floor{n/2}+f+1$ correct servers. For this reason, by the Termination properties of the BRB service, the condition in Line~\ref{code:swbdlo-server-delivery} will eventually be satisfied exactly once and 
    record $r$ is inserted in the local set $S_i$ (Line~\ref{code:swbdlo-server-setUpdate} of Algorithm~\ref{code:ec-server-SWBDLO}). So the condition in Line~\ref{code:swbdlo-server-recordAddWait} of Algorithm~\ref{code:ec-server-SWBDLO} turns true and the response is sent back to the correct client $w$. Since this holds for at least $\floor{n/2}+f+1$ correct servers that received the request, and 
    $\floor{n/2}+f+1 > f+1$, the condition in Line~\ref{code:swbdlo-client-addWait} of Algorithm~\ref{code:ec-client-SWBDLO} will
    be satisfied and the append operation will terminate.
    
    \noindent{\bf\em Byzantine Eventual Consistency}: In order to demonstrate Byzantine Eventual Consistency we need to demonstrate Properties (a) and (b) of Definition \ref{def:ec} with respect to histories $\hist{\mathcal{L}}$ that contain only events of get operations by correct clients and append operations of records that are returned in those get operations. Note that $\mathcal{L}.\act{append}(\rho)$ and $\mathcal{L}.\act{get}()$ are considered in place of $\ledger.\act{add}(r)$ and $\ledger.\act{get}()$. 
    \begin{itemize}
        \item \emph{Property (a):} Let $\mathcal{L}.\act{get}$ be a complete operation in $\hist{\mathcal{L}}$. Let $S$ be the set from where the sequence returned by $\mathcal{L}.\act{get}$ is extracted. Then, from Line~\ref{code:swbdlo-client-get-lcp} of Algorithm~\ref{code:ec-client-SWBDLO}, $\forall r \in S$ the client verified that $r$ belongs to $f+1$ different sets $S_i$ (Line~\ref{code:swbdlo-client-get-f1check} of Algorithm~\ref{code:ec-client-SWBDLO}) returned in a {\sc getResp}($c$, $i$, $S_i$) by different servers. 
        This means that at least a correct server has $r \in S_i$. A server only adds data to its local set $S_i$ if that data was BRB-delivered in {\sc propagate}(-, {\sc add}(-, -, $r$)) messages from $\floor{n/2}+f+1$ different servers. Thanks to the Validity property of the BRB service, this means that at least $\floor{n/2}+f+1$ servers called BRB-broadcast with that message. Again, at least $\floor{n/2}+1$ of them are correct, and they called BRB-broadcast because they received {\sc add}($c$, $p$, $r$) from client $w$. So,
         $\forall r = (k, \rho) \in S$, an $\mathcal{L}.\act{append}(\rho)$ invocation precedes the $\mathcal{L}.\act{get}$ response.
        
        \item {\em Property (b):} This is equivalent to say that $\forall \rho$ such that $\mathcal{L}.\act{append}(\rho) \in \hist{\mathcal{L}}$, eventually there exist a time $t$ such that $\rho$ will be included in all the sequences returned by complete $\mathcal{L}.\act{get} \in \hist{\mathcal{L}}$ invoked after $t$.
        
        Assume $w$ is Byzantine and consider an operation $\mathcal{L}.\act{append}(\rho) \in \hist{\mathcal{L}}$. Then, some $\mathcal{L}.\act{get}()$ operation by a correct client returned a sequence with $r=(k,\rho)$, which means that it received at least $f+1$ messages {\sc getResp}($c$, $i$, $S_i$) in which $r \in S_i$. This means that at least one correct server $i$ had $r \in S_i$. Then, server $i$ BRB-delivered
        at least $\floor{n/2}+f+1$ {\sc propagate}(-, {\sc add}(-, -, $r$)) messages, and by the Termination properties of the BRB service all correct servers $j$ will do as well, and will include $r$ in their local sets $S_j$. After then, any other get operation will always have $f+1$ responses including $r$ from correct servers.
        
        Assume now that $w$ is correct. Then, it sends requests {\sc add}($c$, $w$, $r$) with $r=(k,\rho)$ to at least $\floor{n/2}+2f+1$ servers, so that at least $\floor{n/2}+f+1$ correct ones will process it calling BRB-broadcast in Line~\ref{code:swbdlo-server-broadcast} of Algorithm~\ref{code:ec-server-BGS}. From the Termination properties of the BRB service, $\floor{n/2}+f+1$ {\sc propagate}(-, {\sc add}(-, -, $r$)) messages coming from different servers will be eventually BRB-delivered to all correct servers. Then, all correct servers will eventually add $r$ to their local $S_i$ because of the fulfilment of $\floor{n/2}+f+1$ requirement in Line~\ref{code:swbdlo-server-delivery} of Algorithm~\ref{code:ec-server-BGS}. 
        $\mathcal{L}.\act{get}()$, on its side, returns $r$ if it was seen at least in $f+1$ {\bf out of $2f+1$} different responses. Since at most $f$ can have Byzantine behaviour and eventually all server will include $r$ in their local $S_i$, there will exist a moment in which $\mathcal{L}.\act{get}()$ will always have $f+1$ responses including $r$ from correct servers.
        
        We have shown that, independently of whether $w$ is correct, if $\rho$ is returned in some get operation of a correct client, eventually a record $r=(k,\rho)$ will be in all the sets $S_j$ of all correct servers $j$. Then, there exist a moment in which $r$ is definitely always part of temporary set $A$ in Line~\ref{code:swbdlo-client-get-f1check} of Algorithm~\ref{code:ec-client-SWBDLO} in all get operations.
        Now, in order to ensure that $r$ is part of $S$, and the sequence returned, we need to demonstrate that Line~\ref{code:swbdlo-client-get-lcp} of client Algorithm~\ref{code:ec-client-SWBDLO} does not filter it, eventually. 
        We proceed by induction. If $r.k=1$ then record $r$ is included in $S$. If $r.k > 1$, assume the claim true for record
        $r'=(k-1,\rho')$. I.e., there is a time $t'$ after which $r'$ is always in $A$. Then, there is a time $t \geq t'$ in which
        both $r$ and $r'$ are always in $A$. After $t$ record $r$ will always be included in $S$ and returned by all get operations.
    \end{itemize}
    \noindent{\bf\em Byzantine Strong Prefix}: Let $S = (r_0, ..., r_a)$ and $S' = (r'_0, ..., r'_b)$ the two sets from which the
    sequences returned by the two $\mathcal{L}.\act{get}()$ operations are extracted in Line~\ref{code:swbdlo-client-seq-create}.
    Just as a convenience in notation, we will refer $r = (k, \rho)$ as $r_k = \rho$. Line~\ref{code:swbdlo-client-get-lcp} of the client Algorithm~\ref{code:ec-client-SWBDLO} ensures that records in $S$ and $S'$ can be ordered and that there are not missing element in the sequence. If $S$ and/or $S'$ are empty then one is trivially prefix of the other. So let's assume they both have at least one element and, without loss of generality, that $a \leq b$. 
    Also, let us assume by way a contradiction that the sequence extracted from $S$ is not a prefix of the sequence from $S'$.
    This is equivalent to state that $\exists i \leq a: r_i \neq r'_i$. From Line~\ref{code:swbdlo-client-get-f1check} of Algorithm~\ref{code:ec-client-SWBDLO} we know that $r_j$ and $r'_j$ 
    with $1\leq j \leq a$ were returned at least by one correct server in their respective get operations. 
    So, assuming that such an index $i$ exists means that at least two correct servers executed Line~\ref{code:swbdlo-server-setUpdate} of
    Algorithm~\ref{code:ec-server-SWBDLO}
    for the two records, respectively. This implies that, for both, the condition of Line~\ref{code:swbdlo-server-delivery} was true because they received messages {\sc propagate}($-$, {\sc add}($c$,$p$,$r_i$)) from a set $C$ of at least $\floor{n/2}+f+1$ servers, and messages {\sc propagate}($-$, {\sc add}($c$,$p$,$r'_i$)) from a set $C'$ of at least $\floor{n/2}+f+1$ servers. Note that each $C$ and $C'$
    contains at least $\floor{n/2}+1$
    correct servers. It is obvious that broadcasters of these {\sc propagate} messages must intersect in at least one correct server $j$. So, from the Validity property of the BRB service, at least correct server $j$ called both BRB-broadcast({\sc propagate}($j$, {\sc add}($c$,$p$,$r_i$))) and BRB-broadcast({\sc propagate}($j$, {\sc add}($c$,$p$,$r'_i$))). Line~\ref{code:swbdlo-server-ifnotinT} of Algorithm~\ref{code:ec-server-SWBDLO} filters the received {\sc add}($c$,$p$,$r$) request, so only if $r.k\notin T$
    they are propagated via the BRB-broadcast. If so, Line~\ref{code:swbdlo-server-indexUpdate} adds $r.k$ to $T$ right after the BRB-broadcast. Assume, w.l.o.g., that $j$ received {\sc add}($c$,$p$,$r_i$) before receiving {\sc add}($c$,$p$,$r'_i$). 
    As soon as $j$ BRB-broadcast {\sc propagate}($i$, {\sc add}($c$,$p$,$r_i$))), it added $r_i.k$ to $T$. Then,
    when it received {\sc add}($c$,$p$,$r'_i$) it found that $r'_i.k \in T$, and 
    BRB-broadcast({\sc propagate}($i$, {\sc add}($c$,$p$,$r'_i$))) was not executed.
    But this is a contradiction, and we conclude that our assumption that $\exists i \leq a: r_i \neq r'_i$ is not correct. Hence, the sequence extracted from $S$ must be a prefix of the sequence from $S'$.
\end{proof}

\section{Conclusions and Future Work}
{In this paper we formally define the notion of a Byzantine-tolerant Distributed \gset Object (\BDGSO) and provide client and server algorithms to implement 
a consensus-free eventually consistent \BDGSO.
\remove{the \BDGSO through its sequential specification in Definition~\ref{def:sspec} and then we continue with the Definition~\ref{def:ec} of the eventually consistent \BDGSO. We provide client and server algorithms to implement such a data structure and we demonstrated that they satisfy eventual consistency in Theorem~\ref{proof:ec}.}} 
Then we proceed with some use cases for \BDGSO. {Building on the work in~\cite{BAA2020} and using \BDGSO{}s we provide a consensus-free solution to the \atomicappends problem.}
\remove{that enhance our previous work in \cite{BAA2020}}
Similarly, we provide a consensus-free solution to the Atomic Adds problem, the analogous problem that uses sets instead of ledgers. 
\remove{that solves a similar problem with multiple sets instead of ledgers.}
Finally, we show how a few modifications to the client and server algorithms of \BDGSO, enable to realise an eventual consistent Single-Writer Byzantine Distributed Ledger without solving consensus among servers but still guaranteeing the Byzantine Strong Prefix property. Single-Writer consensus-free BDLO can be suitable for many use cases, like implementing a cryptocurrency or a punch in/out system for employees of a company.
These are scenarios where realising transactional systems in a Byzantine failure model through consensus may not provide reasonable performance, since the need of updating the system global status prevents sustaining a high throughput of operations. Our future plans include implementing and experimentally evaluating the algorithms proposed in this work, as well as specifying a cryptocurrency based on single-writer BDLOs.

\bibliographystyle{acm}
\bibliography{biblio, biblio2, references}

\appendix
\newcommand{\lledger}{\mathcal{L}}

\section*{\LARGE Appendix}

\section{DLO Definitions}
\label{appendix:DLO}
For the reader's convenience, we provide the basic definitions regarding Distributed Ledger Objects~\cite{DLO_SIGACT18}.\newline

\noindent A {\bf\em ledger} $\lledger$ is a concurrent object that stores a totally ordered sequence $\lledger.S$ of records and supports two operations (available to any process $\pr$):
(i) $\lledger.\act{get}_\pr()$, and (ii) $\lledger.\act{append}_\pr(\rdr)$.
The sequential specification of a ledger $\lledger$ is as follows:

\begin{definition}
\label{def:dlo-sspec}
	The \emph{sequential specification} of a ledger $\lledger$ over the sequential history $\hist{\lledger}$ is defined as follows. The value of the sequence $\lledger.S$ of the ledger is initially the empty sequence.
	If at the invocation event of an operation $\op$ in $\hist{\lledger}$ the value of the sequence in ledger $\lledger$ is $\lledger.S=V$, then: 
	\begin{enumerate}
		\item  if $\op$ is an $\lledger.\act{get}_\pr()$ operation, then the response event of $\op$ returns $V$, and
		\item if $\op$ is an $\lledger.\act{append}_\pr(\rdr)$ operation, 
		then at the response event of $\op$, the value of the sequence in ledger $\lledger$ is $\lledger.S=V\extends r$ (where $\extends$ is the concatenation operator).
	\end{enumerate}
\end{definition}

\noindent A {\bf\em Distributed Ledger Object}, DLO for short, is a concurrent ledger object
that is implemented in a distributed manner. In particular, the ledger object is implemented by \emph{servers}, and {\em clients}  invoke the $\act{get}()$ and $\act{append}()$ operations.

\begin{definition}
	\label{def:bdlo-ec}
	A DLO $\lledger$ is {\em\bf eventually consistent} if, given any history $\hist{\mathcal{L}}$,
	\begin{enumerate}
	\item[(a)] Let $S$ be the sequence of records returned by any complete operation $\pi=\act{get}() \in \hist{\mathcal{L}}$ and $\rho_i$ the generic record that belongs to $S$. For each $\rho_i \in S$ then $\hist{\mathcal{L}}$ contains $\act{append}(\rho_j)$ for $j = 1...i$ whose {invocation events appear before the response event of $\pi$ in $\hist{\mathcal{L}}$, and}
		\item[(b)] for every complete operation $\mathcal{L}.\act{append}(\rho) \in \hist{\mathcal{L}}$, there exists a history $\hist{\mathcal{L}}'$ that extends 
		$H_{\mathcal{L}}$ such that,
		for every history $\hist{\mathcal{L}}''$ that extends $H'_{\mathcal{L}}$, every complete operation $\mathcal{L}.\act{get}()$ in $\hist{\mathcal{L}}'' \setminus \hist{\mathcal{L}}'$ returns 
		a sequence that contains $\rho$.  
	\end{enumerate}
\end{definition}

Observe that the above definition is equivalent to the one given in~\cite[Definition 4]{DLO_SIGACT18}.\newline

\noindent A DLO is an {\bf\em eventually consistent Byzantine-tolerant DLO} (BDLO), if it satisfies the next three properties:
\begin{itemize}
\item {\emph{Byzantine Completeness (BC)}: All the \act{get}() and
  {\act{append}()} operations invoked by correct clients eventually
  complete.}
    \item \emph{Byzantine Strong Prefix (BSP)}: If two \textit{correct
      clients} issue two $\act{get}()$ operations that return
      record sequences $S$ and $S'$ respectively, then either $S$ is a
      prefix of $S'$ or vice-versa.
    \item \emph{Byzantine Eventual Consistency (BEC)}: This is the property {of}
     Definition~\ref{def:bdlo-ec} with respect to the $\act{get}()$ operations invoked by correct clients {and the $\act{append}(r)$ operations that append the records $r$ returned in those $\act{get}()$ operations}. 
\end{itemize}

\end{document}